\documentclass[12pt]{article}
\usepackage[latin1]{inputenc}
\usepackage[english]{babel}
\usepackage{amsfonts,amsmath,mathrsfs,amssymb,breqn}
\usepackage[round,comma]{natbib}
\usepackage{graphicx,epsfig}
\usepackage{verbatim}
\usepackage{booktabs}
\usepackage{subcaption}
\usepackage[usenames,dvipsnames]{color}
\usepackage{caption}
\usepackage[lined,boxed,commentsnumbered]{algorithm2e}
\usepackage{multirow}
\usepackage[top=2.5cm,left=2.5cm,right=2.5cm,bottom=2.8cm]{geometry}
\usepackage[bookmarks,colorlinks,breaklinks]{hyperref}
\definecolor{dullmagenta}{rgb}{0.4,0,0.4}   
\definecolor{darkblue}{rgb}{0,0,0.4}
\hypersetup{linkcolor=red,citecolor=blue,filecolor=dullmagenta,urlcolor=darkblue} 

\def\T{{\footnotesize {^{_{\sf T}}}}} 
\newcommand{\Real}{{\rm I}\negthinspace {\rm R}}

\usepackage{color}

\newtheorem{theorem}{Theorem}[section]
\newtheorem{lemma}[theorem]{Lemma}

\newenvironment{proof}[1][Proof]{\begin{trivlist}
\item[\hskip \labelsep {\bfseries #1}]}{\end{trivlist}}

\begin{document}
\bibliographystyle{rss}

\title{Robust approximate Bayesian inference}
\author{Erlis Ruli, Nicola Sartori and Laura Ventura \\
{\it {\small Department of Statistical Sciences, University of Padova, Italy}} \\
{\tt {\small ruli@stat.unipd.it, sartori@stat.unipd.it, ventura@stat.unipd.it}} }
\maketitle

\begin{abstract}
We discuss an approach for deriving robust posterior distributions from $M$-estimating functions using Approximate Bayesian Computation (ABC) methods. In particular, we use \emph{M}-estimating functions to construct suitable summary statistics in ABC algorithms. The theoretical properties of the robust posterior distributions are discussed. Special attention is given to the application of the method to linear mixed models. Simulation results and an application to a clinical study demonstrate the usefulness of the method. An \texttt{R} implementation is also provided in the \texttt{robustBLME} package.
\end{abstract}


\section{Introduction}
\label{sec:intro}

The normality assumption is the usual basis of many statistical analyses in several
fields, such as medicine, health sciences, quality control and engineering statistics. Under this assumption, standard parametric estimation and testing procedures are
simple and efficient. However, both from a frequentist or a Bayesian perspective, it is well known that these
procedures are not robust when the normal distribution is just an
approximate model or in the presence of outliers in the
observed data. In these situations, robust statistical methods can be considered in order to produce statistical procedures that are stable with respect to small changes in the data or to small model departures; see \cite{huber2009robust} for a review on robust methods.

The concept of robustness has been widely discussed in the frequentist literature; see, for instance, \cite{hampel1986,tsou1995robust} and \cite{markatou1998wieg}. Also Bayesian robustness with respect to model misspecification have attracted considerable attention. For instance, \cite{lazar03}, \cite{greco2008robust}, \cite{vcr2010} and \cite{agostinelli2013} discuss approaches based on robust pseudo-likelihood functions, such as the empirical likelihood, as replacement of the genuine likelihood in Bayes' formula. \cite{lewis2014bayesian} discuss an approach for building posterior distributions from robust $M$-estimators using constrained Markov Chain Monte Carlo (MCMC) methods. Recent approaches based on tilted likelihoods can be found in \cite{grunwald2014inconsistency}, \cite{watson2016approximate}, \cite{miller2015robust}. Finally, approaches based on model embedding through heavy-tailed distributions are discussed by \cite{andrade2006bayesian}.

The aforementioned approaches may present some drawbacks. The empirical likelihood is not computable for small sample sizes and posterior distributions based on the quasi-likelihood can be easily obtained only for scalar parameters. The restricted likelihood approach of \cite{lewis2014bayesian}, as well as all the approaches based on estimating equations can be computationally cumbersome with some robust $M$-estimating functions (such as, for instance, those used in linear mixed effects models). The tilted and the weighted likelihood approaches refer to concepts of robustness that are not directly related to the one considered in this paper, which is based on the influence function \citep{hampel1986,huber2009robust}. Finally, the idea of embedding the model in a larger structure has the cost of requiring the elicitation of a prior distribution for the extra parameters introduced. Moreover, the statistical procedures derived under an embedded model are not necessarily robust in a broad sense, since the larger model may still be too restricted. 

Here we focus on the robustness approach based on the influence function  and on the derivation of robust posterior distributions from robust \emph{M}-estimating functions, i.e. estimating equations with bounded influence function (see, e.g., \citealp{huber2009robust}, Chap. 3). In particular, we propose an approach based on Approximate Bayesian Computation (ABC) methods (see, e.g.,\ \citealp{beaumont2002approximate}) using robust \emph{M}-estimating functions as summary statistics. The idea extends results of \cite{ruli2016approximate} on composite score functions to Bayesian robustness. The method is easy to implement and computationally efficient, even when the \emph{M}-estimating functions are potentially cumbersome to evaluate. Theoretical properties, implementation details and simulation results are discussed.
 
The rest of the paper is structured as follows. Section 2 sets the background. Section 3 describes the proposed method and its properties. Section 4 investigates the properties of the proposed method in the context of linear mixed models through simulations and an application to a clinical study. Concluding remarks are given in Section 5.


\section{Background on robust $M$-estimating functions}
\label{sec:background}

Let $y = (y_1,\ldots , y_n)$ be a random sample of size $n$, having independent and identically distributed components, according to a distribution function $F_\theta =F(y;\theta)$, with $\theta \in \Theta \subseteq \Real^d$, $d \geq 1$ and $y\in\mathcal{Y}$. Let $L(\theta)$ be the likelihood function based on model $F_\theta$.

Furthermore, let 
\begin{eqnarray}
\Psi_\theta = \Psi(y;\theta)= \sum_{i=1}^n \psi(y_i;\theta)-c(\theta)\,,\label{eq:Psi}
\end{eqnarray} 
 be an unbiased estimating function for $\theta$, i.e.\ such that $E_\theta (\Psi(Y; \theta)) = 0$ for every $\theta$. In \eqref{eq:Psi}, $\psi(\cdot)$ is a known function, $E_\theta(\cdot)$ is the expectation with respect to $F_\theta$ and the function $c(\cdot)$ is a consistency correction which ensures unbiasedness of the estimating function.

A general \emph{M}-estimator (see, e.g., \citealp{hampel1986}, \citealp{huber2009robust}) is defined as the root $\tilde\theta$ of the estimating equation
$\Psi_\theta = 0$. The class of $M$-estimators is wide and includes a variety of well-known estimators. For example, it includes the maximum likelihood estimator (MLE), the maximum composite likelihood estimator \citep[see, e.g.,][and references therein]{ruli2016approximate} and the scoring rule estimator (see e.g.\ \citealp{dawid2016minimum}, and references therein). Under broad regularity conditions, assumed throughout this paper, an $M$-estimator is consistent and approximately normal with mean $\theta$ and variance
\begin{eqnarray}
K (\theta) = H(\theta)^{-1} J(\theta) H(\theta)^{-\T}
\ , 
\label{var}
\end{eqnarray}
where $H(\theta) = -E_\theta (\partial \Psi_\theta/\partial \theta^{\T})$ and $J(\theta) = E_\theta(\Psi_\theta \Psi_\theta^{\T} )$ are the sensitivity and the variability matrices, respectively. The matrix $G(\theta)=K(\theta)^{-1}$ is known as the Godambe information and the form of $K(\theta)$ is due to the failure of the information identity since, in general, $H(\theta) \neq J(\theta)$. 

The influence function (\emph{IF}) of the estimator $\tilde\theta$ is $\mathrm{\emph{IF}}(x;\tilde\theta,F_\theta) \propto \psi (x;\theta)$ and it measures the effect on the estimator $\tilde\theta$ of an infinitesimal contamination at the point $x$, standardised by the mass of the contamination. A desirable robustness property for $\tilde\theta$ is that its \emph{IF} is bounded (B-robustness), i.e.\ that $\psi(x;\theta)$ is bounded. Note that the \emph{IF} of the MLE is proportional to the score function; therefore, in general, the MLE has unbounded \emph{IF}, i.e.\ it is not B-robust. 
\section{Robust ABC inference} 
\label{sec:bayesian}

One possibility to perform robust Bayesian inference is to resort to  a pseudo-posterior distribution of the form \begin{equation}
\pi_R(\theta|y) \propto \pi(\theta) \, L_R(\theta)
\ ,
\label{postR}
\end{equation} 
where $\pi(\theta)$ is a prior distribution for $\theta$ and $L_R(\theta)$ is a pseudo-likelihood based on a robust $\Psi_\theta$, such as the quasi- or the empirical likelihood. This approach  has two main drawbacks: the empirical likelihood is not computable for very small sample sizes and for moderate sample sizes the corresponding posterior appears to have always heavy tails \citep[see, e.g.,][]{greco2008robust}; moreover, the posterior distribution based on the quasi-likelihood can be easily obtained only for scalar parameters.  A further limitation of this approach is related to computational cost, in the sense that it requires repeated evaluations of the consistency correction $c(\theta)$ in \eqref{eq:Psi}, which in practice is often cumbersome.

We propose an alternative method for computing posterior distributions based on robust $M$-estimating functions, extending the idea in \cite{ruli2016approximate}. The method resorts to the ABC machinery \citep[see, e.g.,][]{beaumont2002approximate} in which a standardised version of $\Psi_\theta$, evaluated at a fixed value of $\theta$, is used as a summary statistic. In \cite{ruli2016approximate} the composite score function is used as a model-based data reduction procedure for ABC in complex models. Here we generalise the approach to general unbiased robust estimating functions. \textcolor{red}{In particular, let $\tilde\theta=\tilde\theta(y)$ be the $M$-estimate of $\theta$ based on the observed sample $y$. Furthermore, let $B_R(\theta)$ be such that $J(\theta)=B_R(\theta) B_R(\theta)^{\T}$.} The summary statistic in ABC is then the rescaled \emph{M}-estimating function

\begin{equation}
\eta_R(y^*;\theta) = B_R(\theta)^{-1} \Psi(y^*;\theta)
 \ ,
\label{etar}
\end{equation}
evaluated at $\tilde\theta$, where $y^*$ is a simulated sample. In the sequel we use the shorthand notation $\tilde{\eta}_R(y^*) = \eta_R(y^*;\tilde{\theta})$.

To generate posterior samples we propose to use the ABC-R algorithm with an MCMC kernel (Algorithm~1), which is similar to Algorithm 2 of \cite{fearnhead2012constructing}; see also \cite{marjoram2003markov}. More specifically, the ABC-R algorithm (Algorithm 1) involves a kernel density $K_h(\cdot)$, which is governed by the bandwidth $h>0$ and a proposal density $q(\cdot\vert\cdot)$; see the Appendix for the implementation details.

\vspace{0.2cm}
\IncMargin{1em}
\begin{algorithm}
  \SetAlgoLined
  \KwResult{A Markov dependent sample $(\theta^{(1)},\ldots,\theta^{(m)})$ from $\pi_R^{ABC}(\theta|\tilde\theta)$}
  \KwData{a starting value $\theta^{(0)}$, a proposal density $q(\cdot|\cdot)$}
  \For {$i = 1 \to m$}{
 draw $\theta^*\,\sim\,q(\cdot|\theta^{(i-1)})$\\
     draw $y^*\,\sim\, F_{\theta^*}$\\
     draw $u\sim U(0,1)$\\
  \eIf{u $\leq\frac{K_h(\tilde\eta_R(y^{*}))}{K_h(\tilde\eta_R(y^{(i-1)}))}\frac{\pi(\theta^{*})q(\theta^{(i-1)}|\theta^{*})}{\pi(\theta^{(i-1)})q(\theta^{*}|\theta^{(i-1)})}$}{
  set $(\theta^{(i)},\tilde\eta_R^{{(i)}})=(\theta^{*},\tilde{\eta}_R(y^{*}))$
}{
 set $(\theta^{(i)},\tilde\eta_R^{(i)})\, =\, (\theta^{(i-1)},\tilde\eta_R(y^{(i-1)}))$
}
}
  \caption[]{\label{alg:abc-mcmc} ABC-R algorithm with MCMC.}
\end{algorithm}
\DecMargin{1.em}
\vspace{0.2cm}

The proposed method gives Markov-dependent samples from the ABC-R posterior
\begin{eqnarray}
\pi_R^{ABC} (\theta|\tilde\theta) = \frac{\int_{\mathcal{Y}^*} \pi(\theta) \, f(y^*;\theta)K_h(\tilde\eta_R(y^*))\,dy^*}{\int_{\mathcal{Y}^*\times\Theta} \pi(\theta) \, f(y^*;\theta)K_h(\tilde\eta_R(y^*))\,dy^*d\theta}
\ .
\label{postABCR}
\end{eqnarray}
While Algorithm~\ref{alg:abc-mcmc} or the use of a kernel in \eqref{postABCR} are not new ideas in the ABC literature, the novelty here is to incorporate in such machinery the robust summary statistic $\tilde{\eta}_R(y^*)$ in order to obtain a simulated sample from a robust posterior distribution. Using similar arguments to \citet[][]{soubeyrand2013}, it can be shown that, for $h \to 0$, $\pi_R^{ABC}(\theta|\tilde\theta)$ converges to $\pi(\theta|\tilde{\theta})$ pointwise \citep[see also][]{blum2010approximate}, in the sense that $\pi_R^{ABC}(\theta|\tilde\theta)$ and  $\pi(\theta|\tilde{\theta})$ are equivalent for sufficiently small $h$. Since in general (\ref{etar}) does not give a sufficient summary statistic, then $\pi(\theta|\tilde\theta)$  differs from $\pi(\theta|y)$ and information is lost by using (\ref{etar}) instead of $y$. However this difference pays off in terms of robustness in inference about $\theta$.

Posteriors conditional on partial information have been extensively discussed in the literature. \cite{soubeyrand2015weak} study the properties of the ABC posterior when the summary statistic is the MLE or the pseudo-MLE derived from a simplified parametric model. An alternative version of the ABC-R algorithm could be based directly on $\tilde{\theta}$, used as the summary statistic and a, possibly rescaled, distance among the observed and the simulated value of the statistic. Apparently, these two versions of ABC, namely the one based on $\tilde{\theta}$ and that based on (\ref{etar}) seem to be treated in the literature as two separate approaches (see, e.g., \citealp{drovandi2015}). However, both alternatives use essentially the same information, i.e. $\tilde{\theta}$, but through different distance metrics. In addition, for small tolerance levels, these two distances converge to zero, and both methods give a posterior distribution conditional on the same statistic $\tilde{\theta}$.
Indeed, let $\tilde{\theta}$ be the summary statistic of the ABC posterior and let the corresponding tolerance threshold $\epsilon$ be sufficiently small and consider the random draw $\theta^{*}$ and its corresponding simulated summary statistics $\tilde\theta^{*}$ taken with the ABC algorithm. Then, by construction $\tilde\theta^{*}$ will be close to $\tilde{\theta}$.
This implies that also $\tilde\eta_R(y^*) = \eta_R(y^*;\tilde{\theta})$ will be close to $\eta_R(y^*;\tilde{\theta}^*)=0$, and hence $\theta^{*}$ is also a sample from the ABC-R posterior which uses the summary statistic $\tilde\eta_R$. 

Nevertheless, the use of $\tilde{\theta}$ as summary statistic requires the solution of $\Psi_\theta=0$ at each iteration of the algorithm, which could be computationally cumbersome. On the contrary, the proposed approach, besides sharing the same invariance properties stated by \cite{ruli2016approximate}, i.e. invariance with respect to both monotonic transformation of the data and with respect to reparameterisations, has the advantage of avoiding computational problems related to the repeated evaluation of $\Psi_\theta$ as shown by the following lemma. 

\setcounter{theorem}{0}

\begin{lemma}
The ABC-R algorithm does not require repeated evaluations of the consistency correction  $c(\theta)$ involved in $\Psi_\theta$, as given by (\ref{eq:Psi}).
\end{lemma}

\begin{proof}
Let $\tilde{\theta}$ be the solution of $\Psi_\theta = 0$, with $\Psi_\theta$ of the form (1). Then, for a given simulated $y^*$ from $F_{\theta^*}$, we have
\begin{eqnarray*}
\tilde\eta_R(y^*) =  B_R(\tilde{\theta})^{-1}(\Psi(y^*;\tilde\theta) - \Psi(y;\tilde\theta))  = \sum_{i=1}^n(\psi(y_i^*, \tilde\theta) - \psi(y_i, \tilde\theta))\,.
\end{eqnarray*}
This implies that $c(\theta)$ is computed only once, at $\tilde{\theta}$.
\end{proof}

\setcounter{theorem}{0}

Theorem~\ref{th:theo1} below shows that the proposed method gives a robust approximate posterior distribution with the correct curvature, even though $\Psi_\theta$, unlike the full score function, does not satisfy the information identity. Here, correct curvature means that asymptotically the robust posterior distribution and its normal approximation have the same covariance matrix, which is the inverse of the Godambe information, i.e. $K(\theta)$.

\begin{theorem}\label{th:theo1}
The ABC-R algorithm with rescaled \emph{M}-estimating function $\tilde\eta_R(y)$ as summary statistic, as $h \to 0$, leads to an approximate posterior distribution with the correct curvature and is also invariant to reparameterisations.
\end{theorem}

\begin{proof}
The proof follows from Theorem 3.2 of \cite{ruli2016approximate}, by substituting the composite estimating equation with the more general \emph{M}-estimating function $\Psi_\theta$. 
\end{proof}

The ABC-R algorithm delivers thus a robust approximate posterior distribution which does not need calibration. On the contrary, for (\ref{postR}) a calibration is typically required.

Theorem~\ref{th:theo2} below shows that the proposed ABC posterior distribution is asymptotically normal. 

\begin{theorem}\label{th:theo2}
Assume the regularity assumptions of \cite{soubeyrand2015weak} and the usual regularity condition on \emph{M}-estimators \citep[Chap. 4]{huber2009robust} are satisfied. Then, for $n \to \infty$ and $h \to 0$, the posterior $\pi_{R}^{ABC}(\theta|\tilde{\theta})$ is asymptotically equivalent to the density of the normal distribution with mean vector $\tilde\theta$ and covariance matrix $K(\tilde\theta)$:
\begin{eqnarray}
\pi_{R}^{ABC}(\theta|\tilde\theta) \, \dot{\sim} \, N_d (\tilde\theta,K(\tilde\theta)) 
\ .
\end{eqnarray}
\end{theorem}

\begin{proof}
The proof follows from Lemma 2 and Theorem 1 in  \cite{soubeyrand2015weak} and from the asymptotic relation between the Wald-type statistic and the score-type statistic, i.e.
\[
\eta_R(y;\theta)^{\T} \, \eta_R(y;\theta) = \Psi_\theta^{\T} J(\theta)^{-1} \Psi_\theta = 
(\tilde\theta - \theta)^{\T} K(\theta)^{-1} (\tilde\theta - \theta) + o_p(1)
\ .
\]

\end{proof}


If $\psi(y;\theta)$ is bounded in $y$, i.e.\ if the estimator $\tilde\theta$ is B-robust, then the ABC-R posterior is resistant with respect to slight violations of model assumptions. 
More precisely, the following theorem shows that the ABC-R posterior inherits the robustness properties of the estimating equation.

\begin{theorem}
If $\psi(y;\theta)$ is bounded in $y$, i.e.\ if the estimator $\tilde\theta$ is B-robust, then asymptotically the posterior mode, as well as other posterior summaries of $\pi_{R}^{ABC}(\theta|\tilde{\theta})$ have bounded IF. 
\end{theorem}

\begin{proof}
From Theorem~\ref{th:theo2}, the asymptotic posterior mode of $\pi_{R}^{ABC}(\theta|\tilde\theta)$ is $\tilde\theta$, which is $B$-robust. Moreover, following results in \cite{greco2008robust}, it can be shown that asymptotic posterior summaries have bounded \emph{IF} if and only if the posterior mode has bounded \emph{IF}.
\end{proof}

\vspace{0.3cm}

\noindent {\bf Example.} We consider an illustrative example in which we compare numerically the ABC-R posterior, with the classical posterior based on the assumed model and the pseudo-posterior (\ref{postR}) based on the empirical likelihood \citep{lazar03,greco2008robust}. Scenarios with data simulated either from the assumed model or from a slightly misspecified model are considered.

Let $F_\theta$ be a location-scale distribution with location $\mu$ and scale $\sigma>0$, and let $\theta = (\mu, \sigma)$. The Huber's estimating function is a standard choice for robust estimation of location and scale parameters. The $M$-estimating function is
$\Psi_\theta = (\Psi_{\mu}, \Psi_{\sigma})$, with
\begin{eqnarray}
\Psi_\mu = \sum_{i=1}^n \psi_{c_1}(z_i)\,\quad\text{and}\quad
\Psi_\sigma = \sum_{i=1}^{n}\left(\psi_{c_2}(z_i)^2 - k(c_2)\right)
\ ,
\label{huber}
\end{eqnarray}
where $z_i = (y_i-\mu)/\sigma$, $i=1,\ldots,n$, $\psi_{c}(z) = \max[-c,\min(c,z)]$ is the Huber $\psi$-function, $c>0$ is a scalar tuning constant which controls the desired degree of robustness of $\tilde{\theta}$, and $k(\cdot)$ is a consistency correction term.  Let $F_\theta$ be the normal distribution $N(\mu, \sigma^2)$ and assume $\mu$ and $\sigma$ a priori independent with $\mu\sim N(0,10^2)$ and $\sigma\sim \text{halfCauchy}(5)$, where $\text{halfCauchy}(a)$ is the half Cauchy distribution with scale parameter equal to $a$. We consider random samples of sizes $n=\{15, 30\}$ drawn from either the normal distribution with $\theta=(0,1)$ and from a contaminated model $(1-\delta)N(0,1) + \delta N(0, \sigma_1^2)$, with $\sigma_1^2>0$. We set the contamination level equal to 10\%, i.e. $\delta=0.1$, and $\sigma_1^2=10$. Moreover, we fix $c_1=1.345$ and $c_2=2.07$, which imply that $\tilde{\mu}$ and $\tilde{\sigma}$ are, respectively, 5\% and 10\% less efficient than the corresponding MLE under the assumed model (see \citealp{huber2009robust}, Chap. 6).

The genuine, e.g.\ the posterior based on the likelihood function of the normal model, and the pseudo-posterior (\ref{postR}) based on the empirical likelihood (EL) are computed by numerical integration. The ABC-R posterior is obtained using Algorithm~1. From the posterior distributions illustrated in Figure~\ref{fig:fig1} we note that, when the data come from the central model (panels (a)-(b)), i.e. for $\delta=0$, all the posteriors are in reasonable agreement, even if the EL posterior behaves slightly worse, especially the marginal posterior of $\sigma$ with $n=15$. When the data are contaminated (panels (c)-(d)), the genuine posterior is less trustworthy as the bulk of the posterior drifts away from the true parameter value (vertical and horizontal straight lines). This is not the case however for the ABC-R posterior which remains centred around the true parameter value. 
We note that in the contaminated case, the ABC-R posterior is the one with smaller variability. This is due to the fact that the ABC-R posterior is not affected by the very outlying observations coming from the contamination component.

\begin{figure}[ht!]
\centering
	\includegraphics[height=0.49\textwidth, width=0.35\textwidth,angle=-90]{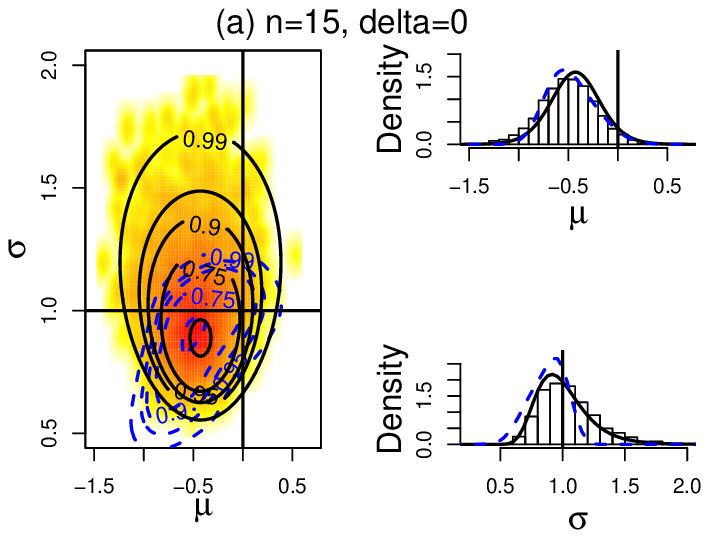}
	\includegraphics[height=0.49\textwidth, width=0.35\textwidth, angle=-90]{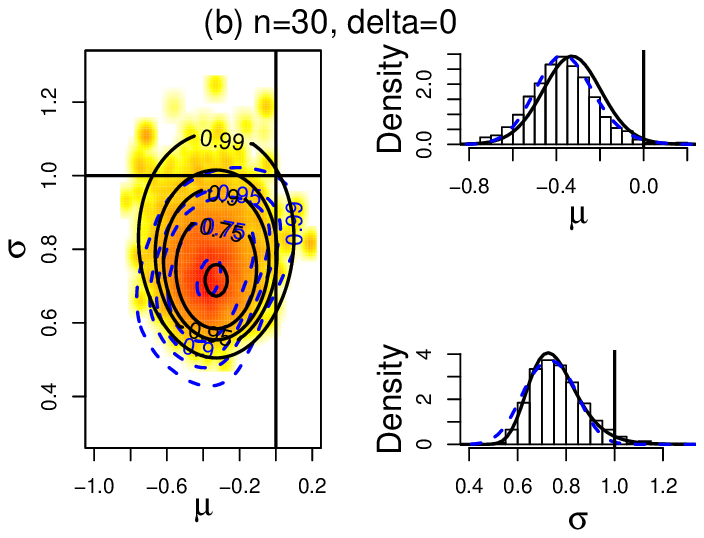}\\
	\includegraphics[height=0.49\textwidth, width=0.35\textwidth, angle=-90]{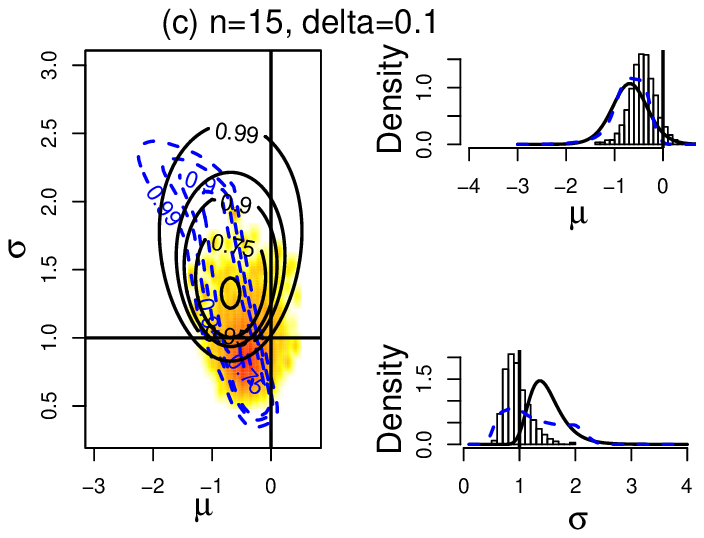}
	\includegraphics[height=0.49\textwidth, width=0.35\textwidth, angle=-90]{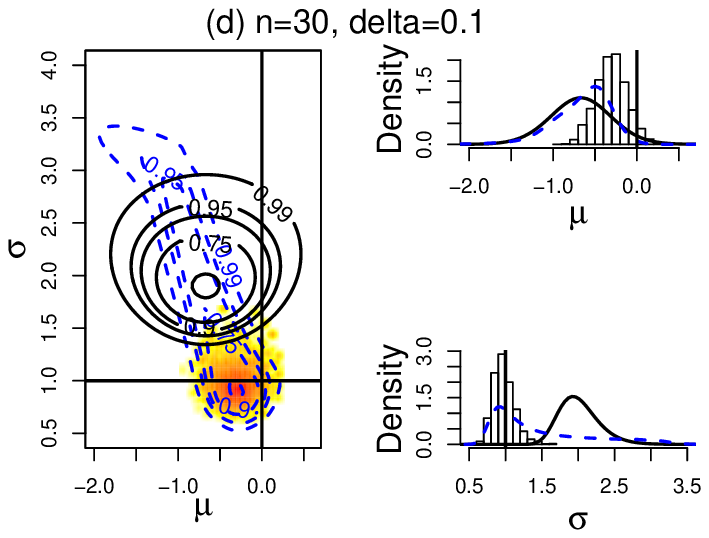}\\
	\vspace{-0.2cm}
	\caption{First row: genuine (black solid), EL (blue dashed) and ABC-R posteriors (shaded image and histogram) for the normal model, when the data come from the central model $N(0,1)$ with (a) $n=15$ and (b) $n=30$. Second row: genuine, EL and ABC-R posteriors for the normal model, when the data come from the contaminated model with $\delta=0.1$, (c) $n=15$ and (d) and $n=30$.}
	\label{fig:fig1}
	\vspace{-0.5cm}
\end{figure}

To highlight the robustness properties of the ABC-R posterior, we consider a sensitivity analysis. A sample $y$ of size $n=31$ is taken from the central model and the aforementioned posteriors are computed from the contaminated data $y^w$ given by the original data with the median observation $y_{(n+1)/2}$ replaced by $y_{(n+1)/2} + w$; $w$ is a contamination scalar with possible values $\{-15,-14,\ldots,15\}.$ The results of the sensitivity analysis, illustrated by means of violin plots in Figure~\ref{fig:fig2}, highlight that the posterior median of the genuine posterior (panel (c))  is substantially driven by $w$. On the other hand, ABC-R and EL posteriors are robust. For all posteriors, the behaviour of the posterior median reflects the behaviour of the \emph{IF} of the posterior mode. Furthermore, the variability of all posteriors is comparable for values of $w$ close to 0. More generally, these plots confirm that the genuine and EL posteriors under contamination are much more dispersed than the ABC-R posterior. 

\begin{figure}[ht!]
\centering
	\includegraphics[width=0.9\textwidth,height=1\textwidth, angle=-90]{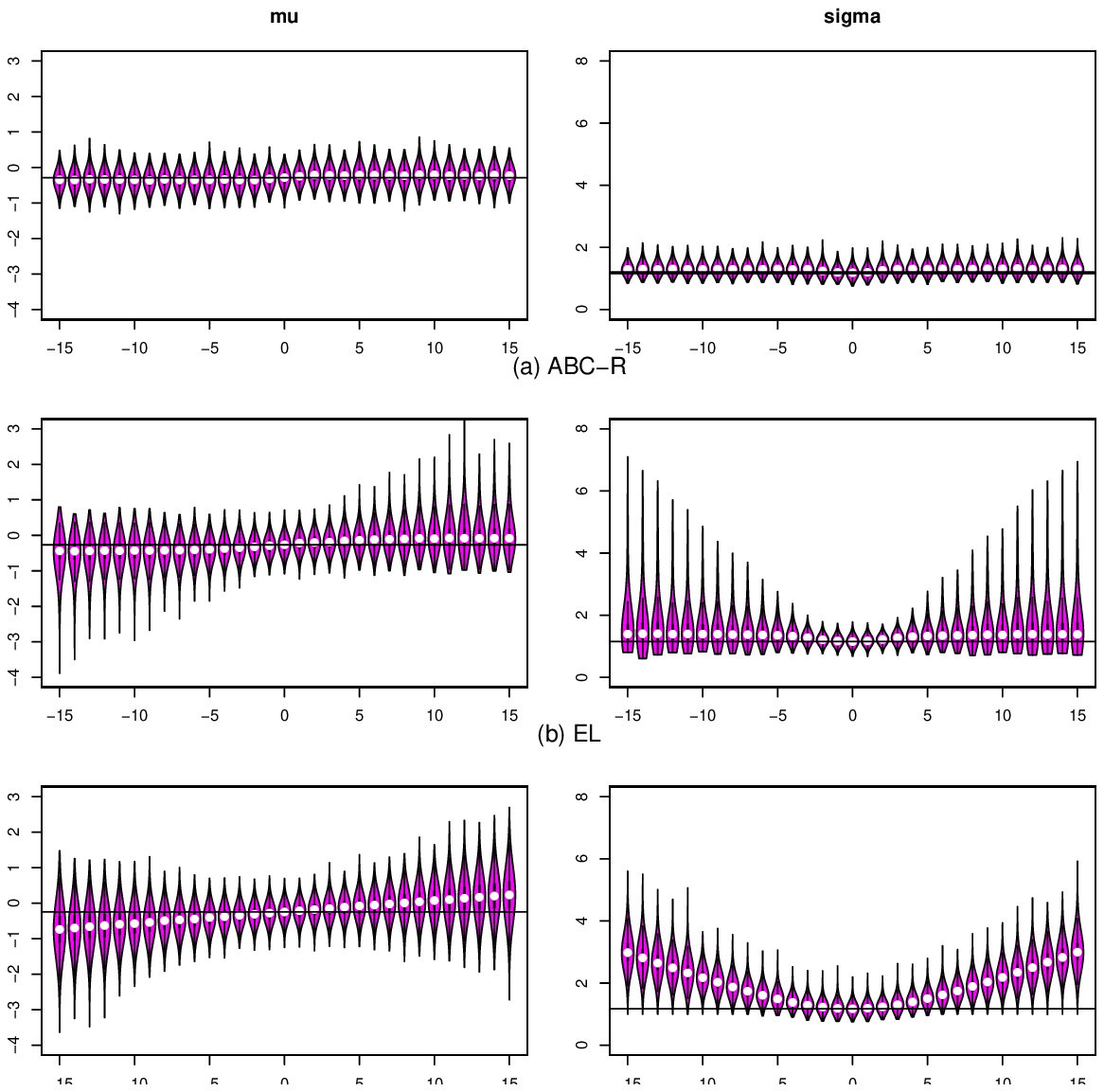}
	\vspace{-0.2cm}
	\caption{Sensitivity analysis for marginal ABC-R (a), EL (b) and genuine (c) posteriors for $\mu$ (left columns) and $\sigma$ (right) represented by means of violin plots. For each violin plot, the central circle represents the posterior median. The horizontal lines denote the corresponding posterior medians under $y^w$ with $w=0$.}
	\label{fig:fig2}
\end{figure}

\section{Application to linear mixed models}
\label{sec:application}

Linear mixed models (LMM) are a popular choice when analysing data in the context of hierarchical, longitudinal or repeated measures. A general formulation  is \begin{equation}
y = X\alpha + \sum_{i= 1}^{c-1} Z_i\beta_i + \varepsilon\,,
\label{eq:model}
\end{equation}
where $y$ is a $n$-dimensional vector of response observations, $X$ and $Z_i$ are known $n\times q$ and $n\times p_i$ design matrices, $\alpha$ is a $q$-vector of unknown fixed effects, the $\beta_i$ are $p_i$-vectors of unobserved random effects $(1\leq i \leq c-1)$ and $\varepsilon$ is a vector of unobserved errors. The $p_i$ levels of each random effect $\beta_i$ are assumed to be independent with mean zero and variance $\sigma_i^2$. Moreover, each random error $\varepsilon_i$ is assumed to be independent with mean zero and variance $\sigma_c^2$ and $\beta_1,\ldots,\beta_{c-1}$ and $\varepsilon$ are assumed to be independent. 

Here we focus on the classical normal LMM, which assumes that $\varepsilon \sim N_n(0_n,\sigma_c^2I_n)$ and $\beta_i \sim N(0, \sigma_i^2)$,  $i=1, \dots, c-1$.  For a normal LMM, it follows that $Y$ is multivariate normal with $E(Y) = X\alpha$ and $\text{var}(Y) = V = \sum_{i = 1}^{c}\sigma_{i}^2Z_iZ_i^\T\,,$ where $Z_{c} = I_n$. We assume that the set of $d = q + c$ unknown parameters $\theta = (\alpha, \sigma^2) = (\alpha, \sigma_{1}^2,\ldots,\sigma_{c}^2)$ is identifiable.  The validity and performance of this LMM
requires strict adherence to the assumed model, which is usually chosen
because it simplifies the analyses and not because it fits exactly the
data at hand. The robust procedure discussed in this paper specifically
takes into account the fact that the normal model is only approximate and
then it produces statistical analyses that are stable with respect to
outliers, deviations from the model or model misspecifications.

Although the $n$ observations $y$ are not independent, if the random effects are nested, then independent subgroups of observations can be found. Indeed, in many situations, $y$ can be split into $g$ independent groups of observations $y_j$, $j = 1, \ldots,g$, and the log-likelihood is
\begin{equation}
\ell(\theta) = \log L(\theta) = - \frac{1}{2}\sum_{j = 1}^g \left\{\log\vert V_j\vert + (y_j-X_j\alpha)^\T V_j^{-1}(y_j - X_j\alpha)\right\}\,,\label{eq:lmmlik}
\end{equation}
where $(y_1,\ldots,y_g)$ and $X$ and $V$ are partitioned accordingly. Classical Bayesian inference for $\theta$ is based on  $\pi(\theta|y)\propto L(\theta) \, \pi(\theta)$, where $\pi(\theta)$ is a prior distribution for $\theta$. However, (\ref{eq:lmmlik}) can be very sensitive to model deviations (\citealp{richardson1995robust}, \citealp{richardson1997bounded}, \citealp{copt2006high}); see also results of the simulation study in Section~\ref{ssec:simulation}.

In the frequentist literature, there are two broad classes of estimators for robust estimation of Gaussian LMM: $M$-estimators \citep[see, e.g.,][and references therein]{richardson1995robust,richardson1997bounded} and $S$-estimators \citep{copt2006high}.
The latter are generally available for balanced designs whereas the formers can be applied to a wide variety of situations; for instance it can deal with unbalanced designs and robustness with respect to the design matrix \citep{richardson1997bounded}. 
In this work we focus on {\emph M}-estimators but it is worth stressing that the idea can be applied to $S$-estimators as well. Following \cite{richardson1995robust}, we focus on the system of {\emph M}-estimating equations
\begin{eqnarray}
&& X^\T V^{-1/2}\psi_{c_1}\left(r\right)=0\,,
\label{eq:richalpha}\\
&&
\psi_{c_2}\left(r\right)^\T V^{-1/2} Z_iZ_i^\T V^{-1/2}  \psi_{c_2}\left(r\right)
- \text{tr}(CPZ_iZ_i^\T)=0,\,i = 1,\ldots,c,
\label{eq:richsigma}
\end{eqnarray}
where $r = V^{-1/2}(y-X\alpha)$ is the vector or scaled marginal residuals, $C = E_\theta \left[\psi_{c_2}(R)\psi_{c_2}(R)^\T \right]$, with $R = V^{-1/2}(Y-X\alpha)$, $P = V^{-1} - V^{-1}X(X^\T V^{-1}X)^{-1}X^\T V^{-1}$ and $\text{tr}(\cdot)$ is the trace operator. The function $\text{tr}(CPZ_iZ_i)$ is a  correction factor needed to ensure consistency at the Gaussian model for each $i=1,\ldots,c$. Equations  (\ref{eq:richalpha})-(\ref{eq:richsigma}) are called robust REML II estimating equations and are bounded versions of restricted likelihood equations. \cite{richardson1997bounded} shows that the $M$-estimator  based on (\ref{eq:richalpha})-(\ref{eq:richsigma}) is asymptotically normal with mean equal to the true parameter $\theta$ and covariance matrix of the form (\ref{var}).  The ABC-R procedure in the normal LMM based on (\ref{eq:richalpha})-(\ref{eq:richsigma}) will be studied by means of simulations in Section~\ref{ssec:simulation} and then applied to a dataset from a clinical study in Section~\ref{sec:antitumour}.

\subsection{Simulation study}
\label{ssec:simulation}

Let us consider the two-component nested model
\begin{equation}
y_{ij} = \mu + \alpha_j + \beta_i + \varepsilon_{ij}\,,\label{eq:modsim}
\end{equation}
where $\mu$ is the grand mean, $\alpha_j$ are the fixed effects, constrained such that $\sum_{j=1}^q \alpha_j=0$, $\beta_i\sim N(0,\sigma_1^2)$ are the random effects and $\varepsilon_{ij}\sim N(0,\sigma_2^2)$ is the residual term, for $j = 1,\ldots,q$ and $i = 1,\ldots,g$. Model (\ref{eq:modsim}) is a particular case of (\ref{eq:model}) with $c = 2$, a single random effect $\beta_1$ with $p_1 = g$ levels and $Z_1$ the unit diagonal matrix. Moreover, the covariate is a categorical variable with $q$ levels; hence the design matrix is given by $q-1$ dummy variables.

We assess the properties of the proposed method via simulations with 500 Monte Carlo replications. For each Monte Carlo replication, the true values for $(\sigma_1^2,\sigma^2_2)$ and for $\alpha$ are drawn uniformly in $(1,10)\times(1,10)$ and $(-5,5)$, respectively. With these values, two datasets of size $g$ are generated: one from the central model and one from the contaminated model $ (1-\delta)N(X_i^\T\alpha,V_i) + \delta N(X_i^\T\alpha, 15V_i)$, where $X_i$ is the matrix of covariates for the $i$th unit, $\theta = (\alpha, \sigma^2_1,\sigma_2^2)$ and $\delta = 0.10$. We consider $q = \{3, 5, 7\}$ and $g = \{30, 50, 70\}$. The prior distributions are $\alpha\sim N_q(0,10^2I_q)$ and $(\sigma_1^2,\sigma^2_2) \sim \text{halfCauchy}(7)\times\text{halfCauchy}(7)$. For each scenario, we fit model (\ref{eq:modsim}) in the classical Bayesian way, using an adaptive random walk Metropolis-Hastings algorithm. The same model is fitted by the ABC-R method using the estimating equations (\ref{eq:richalpha})-(\ref{eq:richsigma}).  As in \cite{richardson1995robust}, we set $c_1 = 1.345$ and $c_2 = 2.07$ and we find $\tilde{\theta}$ solving  (\ref{eq:richalpha})-(\ref{eq:richsigma}) iteratively until convergence. The classical REML estimate, computed by the function \texttt{lmer} of the \texttt{lme4} package, is used as starting value. In our experiments, the convergence of the solution is quite rapid, i.e. $\tilde{\theta}$ stabilises within 10--15 iterations.

We assess the component-wise bias of the posterior median $\tilde\theta_m$ by the modulus of $\tilde{\theta}_m-\theta_0$ in logarithmic scale, where $\theta_0$ is the true value. Moreover, the efficiency of the classical Bayesian estimator relative to the ABC-R estimator is assessed through the index $MD_{MCMC}/MD_{ABC}$, where $MD = \text{med}(|\tilde{\theta}_m-\theta_0|)$; see \cite{richardson1995robust} and \cite{copt2006high}. In addition, for each Monte Carlo replication we compute the Euclidean distance of $\tilde\theta_m$ from $\theta_0$, which can be considered as a global measure of bias. Contrary to \cite{richardson1995robust}, we consider a different $\theta_0$ for each Monte Carlo replication. The bias and efficiency of the classical Bayesian posterior and of the ABC-R posterior for the 500 replications are illustrated in Figures~\ref{fig:fig3} and \ref{fig:fig4}, respectively.

\begin{figure}[t!]
	\centering
\includegraphics[width=1.\textwidth, height = 1.1\textwidth]{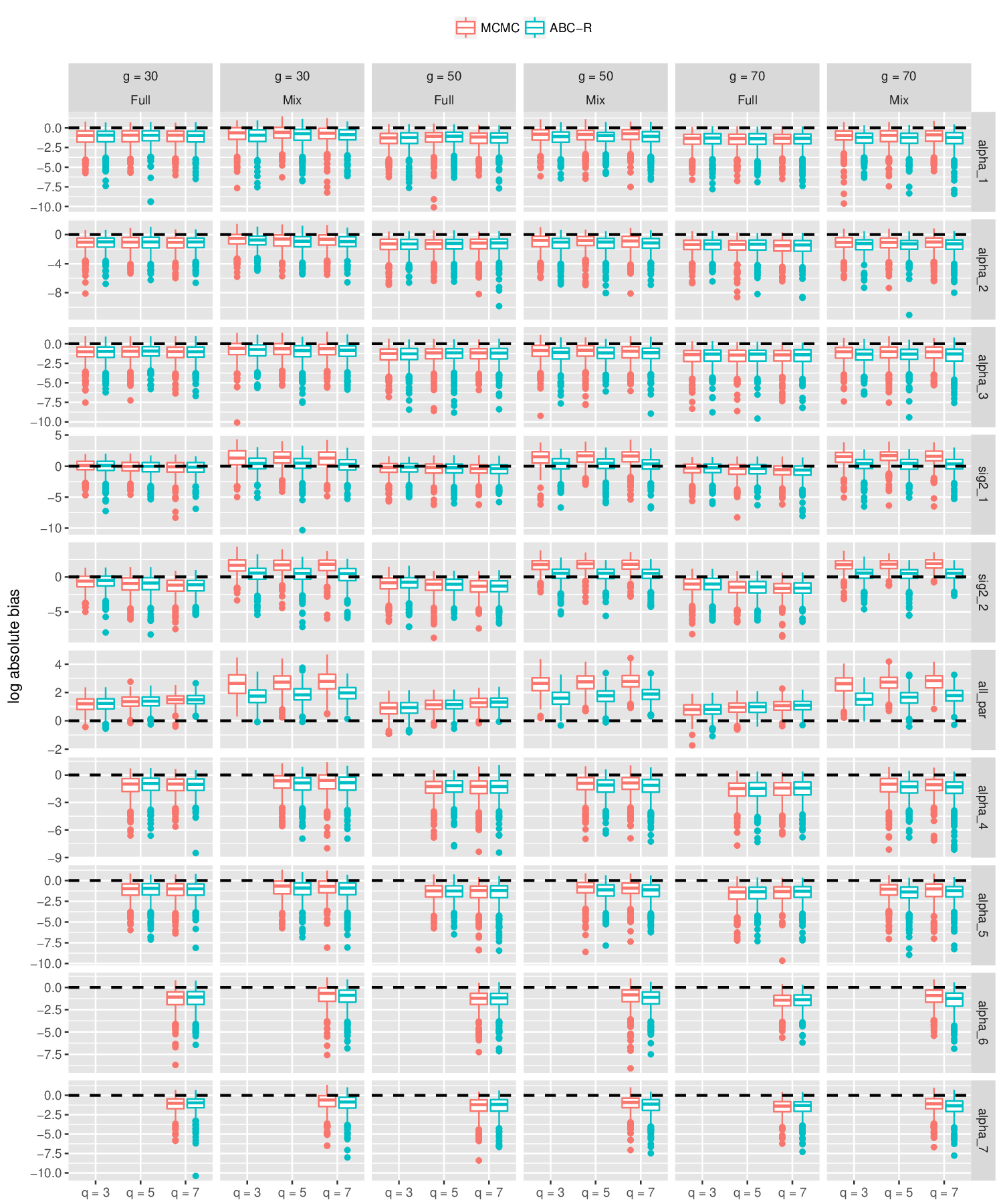}

	\caption{Bias of the ABC-R and classical (\texttt{MCMC}) Bayesian estimation of LMM under either the central (\texttt{Full}) or the contaminated model (\texttt{Mix}) for varying $g$ and $q$. Rows refer to a parameter or combination of parameters (row \texttt{all\_par}); columns within each cell refer to different vales of $q$; e.g. the last two rows (starting from top) have only two boxplots since $\alpha_6$ and $\alpha_7$ are available only with $q = 7$.}
	\label{fig:fig3}
\end{figure}

\begin{figure}[ht!]
	\centering
	\includegraphics[width=1.\textwidth, height = 1.1\textwidth]{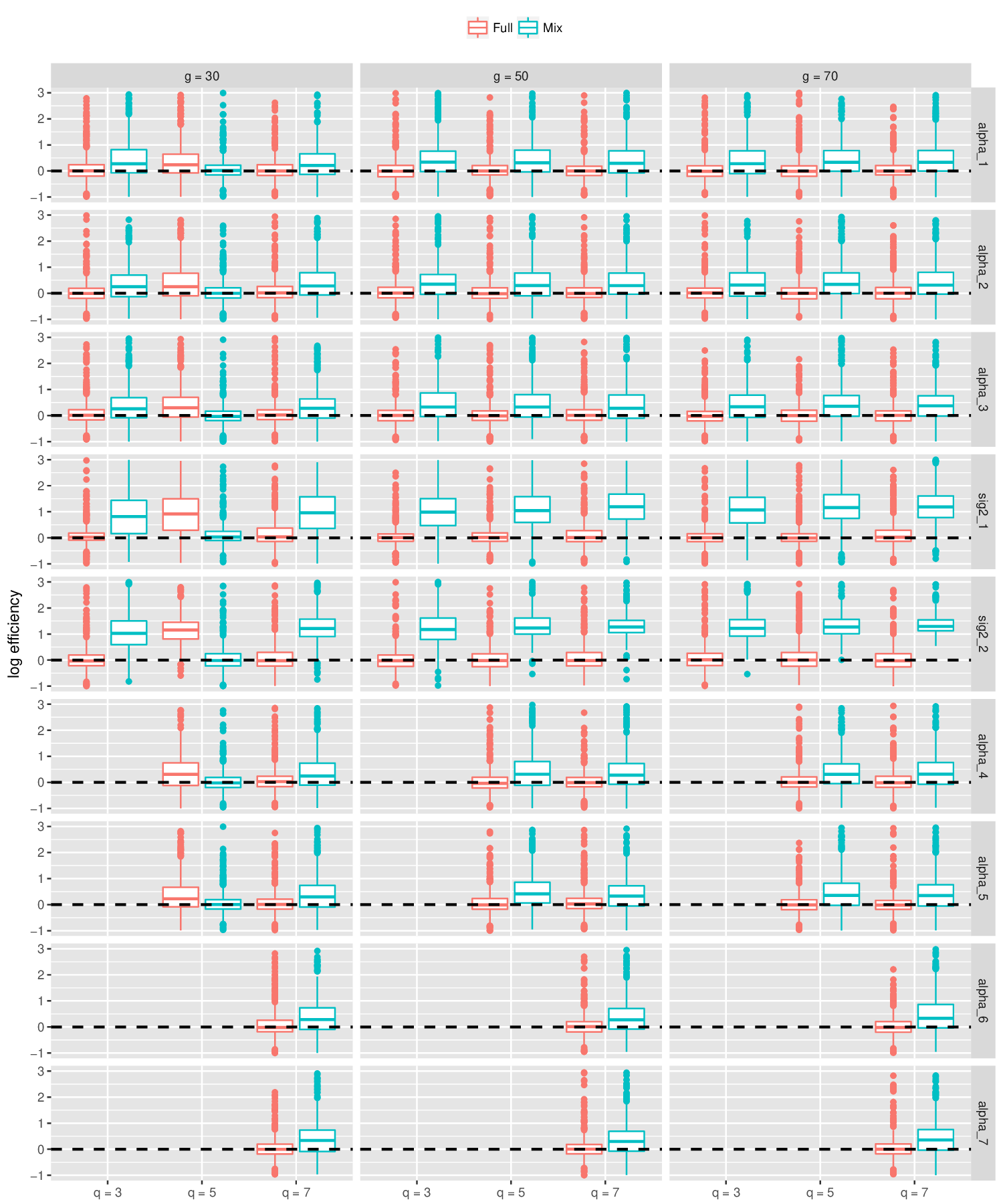}
	\caption{Efficiency of the ABC-R compared to the classical Bayesian estimation of LMM under the central (\texttt{Full}) and the contaminated models (\texttt{Mix}) for varying $g$ and $q$. Rows refer to a parameter and columns within each cell refer to different vales of $q$; e.g. the last two rows (starting from top) have only two boxplots since $\alpha_6$ and $\alpha_7$ are available only with $q = 7$.}
	\label{fig:fig4}
\end{figure}

Under the central model, inference with the ABC-R and the classical Bayesian posteriors is roughly similar, i.e.\ both bias and efficiency compare equally well across the two methods. This holds both for the fixed effects $\alpha$ and for the variance components $(\sigma_1^2,\sigma^2_2)$. Under the contaminated model, we notice important differences among ABC-R and the classical Bayesian estimation. In particular, $\tilde{\theta}_m$ based on ABC-R is less biased, both globally and on a component by component basis, and more efficient. The gain in efficiency is particularly evident for the variance components.

\subsection{Effects of GRP94-based complexes on IL-10}
\label{sec:antitumour}

The \texttt{GRP94} dataset \citep{tramentozzi2016} concerns the measurement of  glucose-regulated protein94 in plasma or other biological fluids and the study of its role as a tumour antigen, i.e. its ability to alter the production of immunoglobines (IgGs) and inflammatory cytokines in the peripheral blood mononuclear cells (PBMCs) of tumour patients. The study involved 27 patients admitted to the division of General Surgery of the Civil Hospital of Padova for ablation of primary, solid cancer of the gastro-intestinal tract. For each patient, gender, age (expressed in years), type and stage of tumour (ordinal scales of four levels) are given. Patients' plasma and PBMCs were challenged with GRP94 complexes and the level of IgG and of the cytokines: interferon$\gamma$ (IFN$\gamma$), interleukin 6 (IL-6), interleukin 10 (IL-10) and tumour necrosis factor $\alpha$ (TNF$\alpha$) were measured. Owing to time and cost constraints, for patients IDs 17, 27 and 28 only IgG was measured. The following five treatments were considered: GRP94 at the dose of either 10 ng/ml or 100 ng/ml, GRP94 in complex with IgG (\texttt{GRP94+IgG}) at the doses 10 ng/ml or 100 ng/ml and IgG a the dose 100 ng/ml. Finally, baseline measurements of IgG and of the aforementioned cytokines were taken from untreated PMBCs. Although fresh patient's plasma and PMBCs are taken for each treatment and patient, the resulting measures are likely to be correlated since plasma and PMBCs are taken from the same patient. Hence, a LMM can be suitable for these data. Using paired Mann-Whitney tests, \cite{tramentozzi2016} show that GRP94 in complex with IgG at the higher dose can significantly inhibit the production of IgG, whereas GRP94 at both doses can stimulate the secretion of IL-6 and TNF$\alpha$ from PBMCs of cancer patients. In addition, some of the differences between treatments were significant for a specific gender; see \cite{tramentozzi2016} for full details.

A feature of these data is the presence of extreme observations, both at baseline and challenged PMBCs-based measurements, as it can be seen from the strip plots in Figure \ref{fig:eda01}. Such extreme observations induce high variability on the response measurements, especially for IFN$\gamma$, IL-6, IL-10 and TNF$\alpha$. Hence, one must be cautious when fitting a LMM to such data. 
\begin{figure}\centerline{
	\includegraphics[width=135mm]{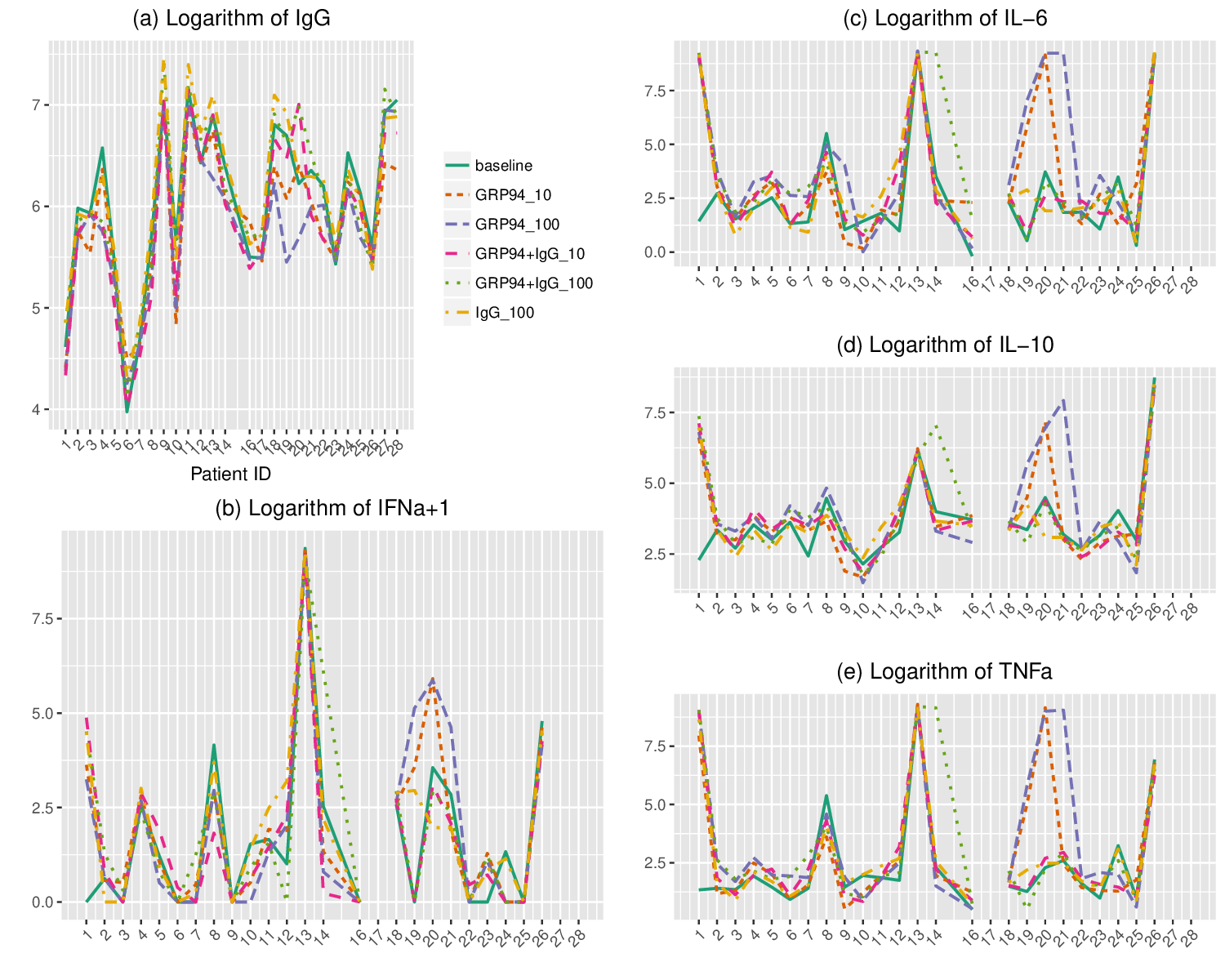}}
	\caption{Strip plots of IgG, IFN$\gamma$, IL-6, IL-10 and TNF$\alpha$ (in logarithmic scale) measured from PBMCs at baseline and after challenging with complexes of GRP94 and IgG. Values on the horizontal axis are (arbitrarily) ordered according to patient ID. Patient ID 15 was removed for clinical reasons and cytokines' measurements for patients with ID 17, 27 and 28 are missing.}
	\label{fig:eda01}
\end{figure}

We fit the two-component nested LMM (\ref{eq:modsim}) to the IL-10 with ABC-R using estimating equations (\ref{eq:richalpha})-(\ref{eq:richsigma}). Since all measures are positive and some of them are highly skewed, a logarithmic transformation is used in order to alleviate distributional skewness. Furthermore, since \cite{tramentozzi2016} highlight a possible gender effect (especially with respect to the cytokines) we also check for gender effects by including an interaction with gender. The model with interaction is
\begin{equation}
y_{i} = X_{i}^\T\alpha + X_i^\T\times\text{w}_i\gamma +\beta_i{1}_6 + \varepsilon_{i}\,,\quad i = 1,\ldots, 24,\label{eq:modinteract}
\end{equation}
where w$_i$ is a dummy variable for gender, $\gamma$ is the fixed effect of the treatment-gender interaction, and ${1_6}$ is the unit vector of dimension 6. The interaction model (\ref{eq:modinteract}) has $12$ unknown fixed effects $(\alpha,\gamma)$.

As in this case there is no extra-experimental information, we assume vague priors. In particular, $\alpha_j\sim N(0,100)$ and $\gamma_j\sim N(0,100)$, for $j = 1,\ldots,6$. For the variance components, following \cite{gelman2006prior}, we assume $\sigma_{1}^2\sim \text{halfCauchy}(7)$ and $\sigma_{2}^2\sim \text{halfCauchy}(7)$ in both models.  However, we note that one of the features of the proposed method is the simultaneous ability to have robustness to possible model misspecification and to include prior information on model parameters, if available.

ABC-R posterior samples are drawn using Algorithm~1. For comparison purposes, we fit also a classical Bayesian LMM with the aforementioned prior and an adaptive random walk Metropolis-Hastings algorithm is used for sampling from this posterior. Figure~\ref{fig:lmmfit} compares the ABC-R and the classical posterior for a subset of the fixed effects of models (\ref{eq:modsim}) and (\ref{eq:modinteract}) by means of kernel density estimations. The parameters shown are those referring to the treatments based on GRP94 at the dose of 10 ng/ml (\texttt{GRP94\_10}), GRP94 at the dose of 100 ng/ml (\texttt{GRP94\_100}) and GRP94 in complex with IgG at the dose of 100 ng/ml (\texttt{GRP94+IgG\_100}), which according to \cite{tramentozzi2016} are the most prominent. The first row (d1) illustrates the marginal posteriors of the parameters of (\ref{eq:modsim}) (with \texttt{baseline} being  the reference category). The second row (d2) shows the marginal posteriors of the parameters of (\ref{eq:modinteract}) (with \texttt{baseline} and \texttt{female} being the reference categories). Numbers within parenthesis in the plot subtitles give the evidence in favour of the null hypothesis H$_0$ that the parameter is equal to zero, computed under the Full Bayesian Significance Testing (FBST) setting of \cite{pereira2008can}; inside the parenthesis, the first (last) value from left refers to the ABC-R (classical) posterior. 

The FBST in favour of $H_0$ has been proposed by \cite{pereira1999evidence} as an  intuitive measure of evidence, defined as the posterior probability related to the less probable points of the parametric space. It favours $H_0$ whenever it is large and it is based on a specific loss function and thus the decision made under this procedure is the action that minimises the corresponding posterior risk \citep{pereira2008can}. The FBST solves the drawback of the usual Bayesian procedure for testing based on the Bayes factor (BF), that is, when the null hypothesis is precise and improper or vague priors are assumed, the BF can be undetermined  and it can lead to the so-called Jeffreys-Lindley paradox.

\begin{figure}[ht!]
\centering
%
%
\includegraphics[width =0.9\columnwidth,height=0.8\columnwidth]{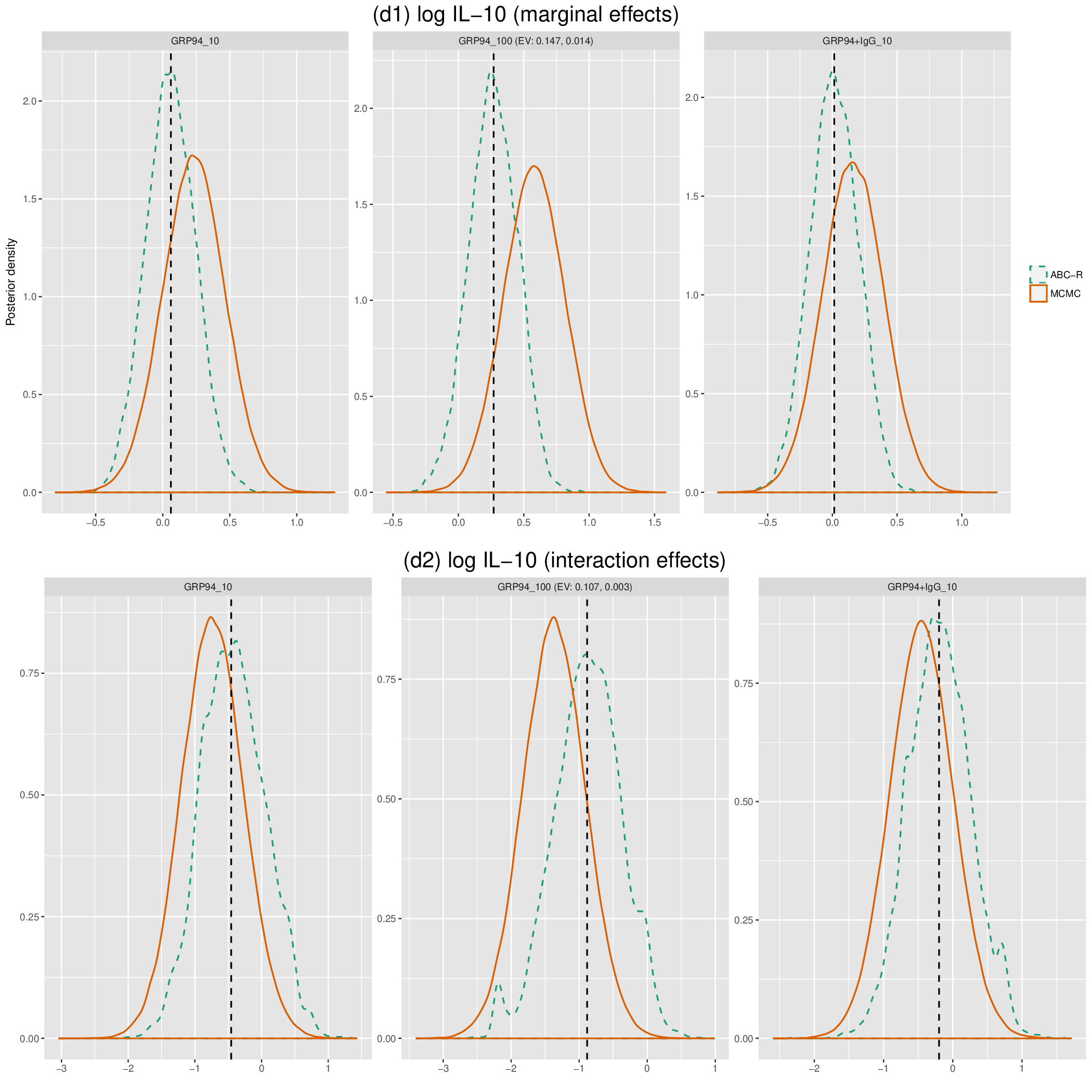}
\vspace{0.2cm}
\caption{Comparison of robust (ABC-R) and full (MCMC) posterior distributions of the fixed effects of the LMM  without interaction with gender (\ref{eq:modsim}) and with interaction (\ref{eq:modinteract}), fitted to the $\log$ IL-10. The first row refers to the posterior of the effects of the treatments against the baseline without interaction; the second refers to the posterior considering interactions of the treatments with gender (with \texttt{baseline} and \texttt{female} being the reference categories). Numbers within parenthesis refer to the FBST evidence in favour of H$_0$ that the parameter is equal to zero; inside the parenthesis, the first (last) value from left refers to the ABC-R (classical) posterior. Dashed vertical lines correspond to components of $\tilde{\theta}$.}
	\label{fig:lmmfit}
\end{figure}


There is a high posterior probability that the effect of \texttt{GRP94\_100}  with or without interaction with gender is different from the baseline, since the evidence of H$_0$ is rather low under the classical Bayesian LMM. However, such effects vanish under the robust ABC-R procedure. This is an indication to the fact that the classical LMM posterior in the case of log IL-10 is likely to be driven by few extreme observations.

\section{Discussion}
\label{s:discuss}
Currently, the only available approach for obtaining posterior distributions explicitly using robust unbiased estimating functions is through pseudo-likelihood methods such as the empirical or the quasi-likelihood \citep{greco2008robust}. \cite{bissiri2016general} show how robust posterior distribution can be based on generic loss functions, in some special cases derived from robust estimating equations. In this work, we present an alternative approach that directly incorporates robust estimating functions into approximate Bayesian computation techniques. With respect to available approaches based on pseudo-likelihoods, our method can be computationally faster when the evaluation of the estimating function is expensive.

Motivated by the \texttt{GRP94} dataset, we focused on two-component nested LMM, but more complex models can be fitted since the estimating equations (\ref{eq:richalpha})-(\ref{eq:richsigma}) are very general \citep[see][]{richardson1997bounded}. For instance, it is possible to deal with models with multiple random effects or even with robustness with respect to the design matrix. An {\tt R} implementation of the proposed method is provided in the {\tt robustBLME} package \citep{robustblme}.

The proposed method can be applied to any unbiased robust estimating equations, such as $S$-estimating equations. The study of the proposed approach with \emph{S}-estimating in the proposed approach is left for future work.

From a practical perspective we recommend to fit both classical and robust LMMs and compare their posteriors, say by FSBT. If the differences are mild then the posterior is probably not impacted by outliers so the classical LMM can be safely used. On the contrary, if there are important differences between them, then it is likely that the LMM posterior is driven by outliers and therefore the robust posterior would be a safer choice.

\section*{Acknowledgements}

This work was partially supported by University of Padova (Progetti di Ricerca di Ateneo 2015, \texttt{CPDA153257}) and by the Italian Ministry of Eduction under the PRIN 2015 grant (\texttt{2015EASZFS\_003}).

\section*{Appendix: Computational details}

Provided simulation from $F_\theta$ is fast, the main demanding requirement of the proposed method is essentially the computation of the observed $\tilde{\theta}$ and the scaling matrix $B_R(\theta)$ evaluated at $\tilde{\theta}$. Given that, for large sample sizes, 
$$
\eta_R (y;\theta) \sim N_d(0_d, I_d)\,,
$$ 
where $0_d$ is a $d$-vector of zeros and $I_d$ is the identity matrix of order $d$, it is reasonable to replace $K_h(\cdot)$ with the multivariate normal density centred at zero and with covariance matrix $hI_d$. In order to choose the bandwidth $h$ we consider several pilot runs of the ABC-R algorithm for a grid of $h$ values, and select the value of $h$ that delivers approximately 0.1\% acceptance ratio (as done, for instance, by \citealp{fearnhead2012constructing}). 

Contrary to other ABC-MCMC algorithms in which the proposal requires pilot runs (see, \citealp{cabras2015}, for building proposal distributions in ABC-MCMC), in our case a scaling matrix for the proposal $q(\cdot|\cdot)$ can be readily build, almost effortlessly, by using the usual sandwich formula (2) evaluated at $\tilde\theta$ (see also \citealp{ruli2016approximate}). Even in cases in which $H(\theta)$ and $J(\theta)$ are not analytically available, they can be straightforwardly estimated via simulation. Indeed, in our experience, 100-500 samples from the model $F_{\tilde\theta}$, give estimates with reasonably low Monte Carlo variability (see also \citealp{cattelan2015empirical}). Throughout the examples considered we use the multivariate $t$-density with 5 degrees of freedom as the proposal density $q(\cdot|\cdot)$ and the ABC-R is always started from $\tilde{\theta}$.  In the ABC algorithm, we fix the tolerance threshold in order to give a pre-specified but small acceptance ratio, as frequently done in the ABC literature.  In our experimentations we found that an acceptance value of 0.1\% gives satisfactory results.

\bibliography{mybibilo}


\end{document}